\documentclass[10pt,conference]{IEEEtran}
\IEEEoverridecommandlockouts

\usepackage[cmex10]{amsmath}
\usepackage{array}
\usepackage{amssymb}
\usepackage{graphicx}

\newtheorem{theorem}{Theorem}

\newtheorem{remark}{Remark}
\newtheorem{proof}{Proof}
\newtheorem{corol}{Corollary}

\newcommand{\Lc}{\mathcal{L}}

\newcommand{\Xc}{\mathcal{X}}
\newcommand{\Yc}{\mathcal{Y}}

\newcommand{\Rc}{\mathcal{R}}

\newcommand{\Sc}{\mathcal{S}}
\newcommand{\Wc}{\mathcal{W}}

\newcommand{\wyd}{w_{12}}
\newcommand{\wds}{w_{23}}
\newcommand{\wys}{w_{13}}

\newcommand{\ryd}{R_{12}}
\newcommand{\rds}{R_{23}}
\newcommand{\rys}{R_{13}}

\begin{document}

\title{State-Dependent Relay Channel with Private Messages with  Partial Causal and Non-Causal Channel State Information}

\author
{Bahareh Akhbari, Mahtab Mirmohseni, and Mohammad Reza Aref\\
Information Systems and Security Lab (ISSL)\\
Electrical Engineering Department, Sharif University of Technology, Tehran, Iran.\\
Email: b\_akhbari@ee.sharif.edu, mirmohseni@ee.sharif.edu, aref@sharif.edu  
\thanks{This work was partially supported by Iranian
National Science Foundation (INSF) under contract No.
84,5193-2006 and by Iran Telecommunication Research Center (ITRC) under contract No. T500/20958.}}
\maketitle
\begin{abstract}
In this paper, we introduce a discrete memoryless State-Dependent Relay Channel with Private Messages (SD-RCPM) as a generalization of the state-dependent relay channel. We investigate two main cases: SD-RCPM with \textit{non-causal} Channel State Information (CSI), and SD-RCPM with \textit{causal} CSI. In each case, it is assumed that \textsl{partial} CSI is available at the source and relay. For non-causal case, we establish an achievable rate region using Gel'fand-Pinsker type coding scheme at the nodes informed of CSI, and Compress-and-Forward (CF) scheme at the relay. Using Shannon's strategy and CF scheme, an achievable rate region for  causal case is obtained. As an example, the Gaussian version of SD-RCPM is considered, and an achievable rate region for Gaussian SD-RCPM with non-causal perfect CSI only at the source, is derived. Providing numerical examples, we illustrate the comparison between achievable rate regions derived using CF and Decode-and-Forward (DF) schemes.
\end{abstract}

\IEEEpeerreviewmaketitle
\vspace{-3mm}
\section{Introduction}
\fontsize{9.6}{11}
\selectfont
The Relay Channel (RC) \cite{CoveElg79}, is a communication system in which a message is transmitted from the source to the destination with the help of a relay. The Partially Cooperative Relay Broadcast Channel (PC-RBC) studied in \cite{LianVee07, LianKra07}, is a generalization of RC in which the source node also sends a private message intended for the relay node. Therefore, in this channel the relay is also a sink of data. The Relay Channel with Private Messages  (RCPM) studied in \cite{TannNos07}, is a generalization of RC in which the relay is both a source and a sink of data. Hence, RCPM  can be regarded as a generalization of PC-RBC  wherein the relay is also a source of data and sends a private message to the destination. The  RCPM model  fits  networks in which  dedicated relays are not available, and relaying is performed by the nodes that each node is a source and a sink of data. 

In this paper, we assume that RCPM is controlled by random parameters called  channel state and we refer to it as State-Dependent Relay Channel with Private Messages (SD-RCPM).
Recently, state-dependent channels have attracted considerable attention. In these channels, the information on the channel state can be known to the terminals causally or non-causally.  For a comprehensive overview on state-dependent channels see  \cite{Jafa06, KeshSteMer08}. 

Among state-dependent multiuser models, some results have also been obtained for the State-Dependent RC (SD-RC) and State-Dependent PC-RBC (SD-PCRBC) \cite{SiguKim05}-\cite{ZaidVan07} for Gaussian and discrete memoryless cases with causal or non-causal Channel State Information (CSI).
Since RCPM can be regarded as a generalization of RC and PC-RBC, the SD-RCPM can be regarded as a generalization of SD-RC and SD-PCRBC.

In this paper, after introducing SD-RCPM, we investigate two main cases: SD-RCPM with non-causal CSI, and SD-RCPM with causal CSI. In each case, we investigate SD-RCPM with partial CSI at the source and the relay. As special cases it includes three different situations in which perfect CSI is available i) only at the source, ii) only at the relay, and iii) both at the source and the relay. Situations i and ii called \textit{asymmetric} scenarios refer to the cases in which CSI is available only at some of the nodes. Situation iii is called the \textit{symmetric} scenario. In fact in the asymmetric cases, it is assumed that only some of the nodes have the ability or the permission to know CSI.  In \cite{AkhbMirAre09, MirmAkhAre09}, unlike most of the previous works  on SD-RC, we have used CF strategy \cite{CoveElg79} for relaying in SD-RC, and showed  that there exist cases (similar to the classic RC) for which  CF scheme achieves  rates higher than those  derived using DF strategy \cite{CoveElg79} for SD-RC with asymmetric CSI in \cite{ZaidKotLan08b, ZaidVan08b}. Now, in this paper we also focus on  CF relaying scheme. In fact similar to what is shown for classic RC, the CF and DF are both able to outperform each other under different conditions. In  DF (or Partial DF)  scheme, the relay has to decode the whole message intended for the destination (or a part of it). Therefore, when only the source is informed of CSI, due to lack of knowledge of CSI at the relay,  the rate loss is caused at the relay as shown in \cite{ZaidShaPiaVan10}. On the other hand, when only the relay is informed of CSI, the source cannot use CSI and cannot conceive what the relay exactly sends, and as shown in \cite{ZaidKotLan08b}, this causes a loss in the  coherence gain  which is  expected to be achieved by DF relaying. However, DF based rate can be optimal for certain cases as shown recently in \cite{ZaidKotLan08b, ZaidShaPiaVan10}. On the other hand, in  CF strategy  independent codebooks are exploited at the source and relay,  and the relay simply  compresses its received signal. Also, CF outperforms DF  for relaying the message intended for the destination, when the link between the source and relay is worse than the direct link.  

Here, for the non-causal situation, we derive an inner bound on the capacity region (achievable rate region) of SD-RCPM based on using Gel'fand-Pinsker (GP) type coding \cite{GelfPin80} at the nodes informed of CSI, and using CF scheme at the relay. We derive an achievable rate region for SD-RCPM in the causal case using Shannon's strategy \cite{Shan58} and CF scheme. We show that our result for the causal case can be considered as a special case of non-causal CSI, and this is in analogy with the relation between the capacity of the single user channel with causal CSI \cite{Shan58}, and its non-causal counterpart \cite{GelfPin80}. We also show that our results subsume the results in \cite{AkhbMirAre09, MirmAkhAre09} which are for SD-RC, and the result in \cite{TannNos07} which is for RCPM (state-independent), as special cases. As an example, we consider the Gaussian version of SD-RCPM with additive independent identically distributed (i.i.d) state process, and  obtain achievable rate region for the case where perfect CSI is available non-causally only at the source. We also illustrate  the trade-off between relayed message rate and private messages rates for this channel. Since there is no available  achievable region for SD-RCPM which is established based on DF scheme, to compare DF and CF strategies for this channel,  we consider a scenario in which the relay sends no private message to the destination (i.e., SD-RCPM reduces to SD-PCRBC). For this scenario, we  make a comparison between our CF based derived achievable region and the  rate region  derived in \cite{ZaidVan07} based  on DF.

The rest of paper is organized as follows. Section \ref{sec:definition}, introduces SD-RCPM  channel model and  notations. Section \ref{sec:noncausal}, investigates SD-RCPM with non-causal CSI. The causal case is considered in section \ref{sec:causal} and finally, Section \ref{sec:Gaussian} contains Gaussian examples. 
\vspace{-2mm}
\section{Preliminaries and Definitions}\label{sec:definition}
In this paper, upper case letters (e.g., $X$) are used to denote Random Variables (RVs) while their realizations are denoted by lower case letters (e.g., $x$). $X_i^j$ indicates the sequence of RVs $(X_i,X_{i+1},\ldots, X_j)$, and for brevity, $X^j$ is used instead of $X_1^j$. $p_X(x)$ denotes the probability mass function (p.m.f) of $X$ on a set $\Xc$, where occasionally subscript $X$ is omitted. $A_\epsilon^n(X,Y)$ denotes the set of strongly jointly $\epsilon$-typical length-$n$ sequences  on $p(x,y)$, which is abbreviated by $A_\epsilon^n$  if it is clear from the context. 

A discrete memoryless SD-RCPM is denoted by $(\Xc_1 \times \Xc_2,p(y_2,y_3|x_1,x_2,s),\Yc_2 \times \Yc_3)$, where $p(y_2,y_3|x_1,x_2,s)$ is the    probability transition function. $X_1\in \Xc_1$ and $X_2\in \Xc_2$ are the source and  relay inputs, respectively. $Y_2\in \Yc_2$ and $Y_3\in \Yc_3$ are respectively the outputs at the relay and  destination, and $s$ denotes the channel state. We assume that the source and relay know an i.i.d noisy version of states (i.e., partial CSI) drawn by $p(s,s_1,s_2)$ where $s\in \Sc$, $s_1\in \Sc_1$ and  $s_2\in \Sc_2$. The CSI at the source (or the relay) is perfect if $s_{1,j}$ (or $s_{2,j}$) equals $s_j$ for  
$1\leq j \leq n$.

A $((2^{nR_{12}},2^{nR_{23}},2^{nR_{13}}),n)$ code for SD-RCPM consists of three message sets $\Wc_{12}=\{1,\ldots,2^{nR_{12}}\}$ (source to relay private message), $\Wc_{23}=\{1,\ldots,2^{nR_{23}}\}$ (relay to destination private message), and $\Wc_{13}=\{1,\ldots,2^{nR_{13}}\}$ (sent from the source to the destination  with the help of the relay),
where  independent messages $W_{12}$, $W_{23}$ and $W_{13}$  are uniformly distributed over respective sets.  It also consists of  an encoder at the source and a set of encoding functions at the relay where for non-causal CSI are defined as $\phi_{1}:\Wc_{12}\times\Wc_{13}\times\Sc_1^{n}\rightarrow\Xc_{1}^{n}$, and   $\phi_{2,j}:\Yc_{2}^{j-1}\times\Wc_{23}\times\Sc_2^{n}\rightarrow\Xc_{2}$ for $j=1,\ldots,n$, respectively. For causal CSI these encoders are respectively defined by 
$\phi_{1,j}:\Wc_{12}\times\Wc_{13}\times\Sc_1^{j}\rightarrow\Xc_{1}$ and $\phi_{2,j}:\Yc_{2}^{j-1}\times\Wc_{23}\times\Sc_2^{j}\rightarrow\Xc_{2}$, for $j=1,\ldots,n$. Two decoding functions $d_1$ and $d_2$ at the relay and destination are defined  respectively as: $d_1:\Yc_2^n \times\Sc_2^{n}\rightarrow \Wc_{12}$ and
$d_2:\Yc_3^n \rightarrow \Wc_{23} \times\Wc_{13}$. Note that, for decoding at the relay there is no difference between causal and non-causal cases since the relay can wait until the end of the block, before decoding. The average probability of error ($P_e^{(n)}$) is defined as the one for RCPM in \cite{TannNos07}. A rate tuple $(R_{12},R_{23},R_{13})$ is said to be achievable for SD-RCPM, if there exists a sequence of codes $((2^{nR_{12}},2^{nR_{23}},2^{nR_{13}}),n)$ with  $P_e^{(n)}\rightarrow 0$ as $n\rightarrow \infty$. 
%\vspace{-2mm}
\section{SD-RCPM with Non-Causal CSI}\label{sec:noncausal}
In this section, we consider SD-RCPM with partial non-causal CSI  at the source and the relay. In the following, the achievable rate region  of this channel is derived.
\begin{theorem}\label{th:noncausal}
In the discrete memoryless SD-RCPM with non-causal CSI $S_1$ and $S_2$ respectively at the source and the relay,  the nonnegative rate tuples  {\small $(R_{12},R_{23},R_{13})$} denoted as $\Rc_1$, satisfying:
\begin{align}
&\!\!R_{13}< I(T_1;\hat Y_2,Y_3|K_2,Q_2)-I(T_1;S_1)\label{eq:lowerCFsourcerelaync1}\\
&\!\!R_{12}< I(T_2;Y_2,S_2|K_2,Q_2)-I(T_2;S_1)\label{eq:lowerCFsourcerelaync2}\\
&\!\!R_{13}+R_{12}< I(T_1;\hat Y_2,Y_3|K_2,Q_2)+I(T_2;Y_2,S_2|K_2,Q_2)\nonumber\\
&\ \ \ \ \ \ \ \ \ \ \ \ \ \ \ \ \ \ \ \ \ \ -I(T_1;S_1)-I(T_2;S_1)-I(T_1;T_2|S_1)\label{eq:lowerCFsourcerelaync3}\\
&\!\!R_{23}<I(K_2;Y_3)-I(K_2;S_2)\label{eq:lowerCFsourcerelaync4},
\end{align}
subject to the constraint
{\small
\begin{flalign}
\!\!\!\! I(\hat Y_2;Y_2,S_2,T_2|K_2,Q_2,Y_3) \leq I(Q_2;Y_3|K_2)-I(Q_2;S_2|K_2),\label{eq:condCFsourcerelaync}
\end{flalign}
}are achievable for any joint p.m.f of the form
\begin{flalign}
&\!\!p(s,s_1,s_2,k_2,q_2,t_1,t_2,x_1,x_2,\hat y_2,y_2,y_3)=\nonumber \\
&\!\!p(s,s_1,s_2)p(k_2|s_2)p(q_2|k_2,s_2)p(x_2|q_2,k_2,s_2)p(t_1,t_2|s_1)\nonumber\\
&\!\!\!\!\times p(x_1|t_1,t_2,s_1)p(y_2,y_3|x_1,x_2,s)p(\hat y_2|y_2,q_2,k_2,s_2,t_2).\label{eq:pmfsourcerelaync}
\end{flalign}
\end{theorem}
\begin{remark}\label{mk:CFsourcerelaync1}
The relay which is the middle terminal, knows CSI $S_2$. Hence, the relay  acting as a decoder of the source-relay link and based on its knowledge of CSI, tries to cancel the effect of the channel state on $y_2$. Moreover, it can compress $S_2$ (besides $Y_2$) and sends it to  destination to provide it with a partial CSI, when needed. For example, if {\small $Y_2=\emptyset$} the relay can  compress only $S_2$  and sends it to the destination, and so the destination can use this partial CSI. To achieve these two goals, in general, we assume that in  \eqref{eq:pmfsourcerelaync} $\hat y_2$ is conditioned on $s_2$. 
Moreover,  $T_2$ which is decoded at the relay is not decoded at the destination ($T_2$ carries source to relay private message), and hence can be regraded as a channel state for destination. So, {\small $T_2$} can be treated similar to {\small $S_2$}, and we assume that in \eqref{eq:pmfsourcerelaync} $\hat y_2$ is also conditioned on $t_2$. Note that, by these assumptions, the case where the relay only compresses {\small $Y_2$} is also included.
\end{remark}

\textit{Outline of the Proof:}
The proof is based on random coding scheme which combines GP-type coding at the source and relay, and CF strategy. More precisely, 
since the source is informed of CSI $S_1$,  the coding scheme  which extends Marton's region to the state-dependent Broadcast Channel (BC) with non-causal CSI at Transmitter (CSIT) \cite{SteiSha05} is used at the source.  The relay decoder uses $S_2$ as part of channel
output. Moreover, the relay encoder uses CF and since it is informed of $S_2$, it
uses GP coding to send the index of the compressed signal which is superimposed on the relay-destination private message.
The auxiliary RV $K_2$ stands for private message sent from the relay to destination using GP coding,
and $Q_2$ represents the bin of compressed signal's index. 
Now, consider a  block Markov encoding scheme where a sequence of  $B-1$ messages  ({\small $W_{12,i},W_{23,i},W_{13,i}$}) for $i=1,\ldots,B-1$ is transmitted in $B$ blocks, each of $n$ symbols.  As {\small $B\rightarrow \infty$}, the  rate tuple {\small $(R_{12},R_{23},R_{13})\times\frac{(B-1)}{B}$} approaches {\small $(R_{12},R_{23},R_{13})$}. 

\textit{Random Coding:}
For any joint p.m.f defined in (\ref{eq:pmfsourcerelaync}), generate  $2^{n(\rys+\rys')}$ i.i.d codewords $t_1^n(\wys,k)$   where $\wys\in[1,2^{n\rys}]$ and $k\in[1,2^{n\rys'}]$. Generate $2^{n(\ryd+\ryd')}$ i.i.d $t_2^n(\wyd,l)$, $\wyd\in[1,2^{n\ryd}]$, $l\in[1,2^{n\ryd'}]$. Generate $2^{n(\rds+\rds')}$ i.i.d codewords $k_2^n(\wds,m)$, $\wds\in[1,2^{n\rds}]$, $m\in[1,2^{n\rds'}]$. 
For each $k_2^n(\wds,m)$, generate $2^{n(R_2+R'_2)}$ i.i.d $q_2^n(t,r|\wds,m)$, $t\in[1,2^{nR_2}]$, $r\in[1,2^{nR_2'}]$.
In codewords $t_1^n$, $t_2^n$, $k_2^n$ and $q_2^n$, the first index represents the bin, and the second one indexes  the sequence within the particular bin. For each $k_2^n(\wds,m)$ and $q_2^n(t,r|\wds,m)$, generate $2^{n\hat R_2}$ i.i.d $\hat y_2^n(z|t,r,\wds,m)$ each with probability
$\prod_{j=1}^n p(\hat y_{2j}|q_{2j},k_{2j})$. Randomly partition the set $\{1,\ldots,2^{n\hat R_2}\}$ into $2^{nR_2}$ bins defined as $B(t)$.

\textit{Encoding (at the beginning of block $i$):}
We assume that the CSI $S_1^n$ and $S_2^n$
in each block are non-causally known to the source and the relay, respectively. Then:

1) Let $(w_{13,i},w_{12,i})$ be the new message pair  to be sent  from the source in block $i$. The source looks for the smallest $k\in[1,2^{n\rys'}]$ and $l\in[1,2^{n\ryd'}]$, such that
$(t_1^n(w_{13,i},k),t_2^n(w_{12,i},l),s_1^n(i))\in A_\epsilon^n$.
Denote this $k$ and $l$ with $k_i$ and $l_i$, respectively. If no such indices $k$ and $l$ exist, an encoding error is declared.  There exist such indices $k_i$ and $l_i$ with arbitrarily high probability, if $n$ is large enough and
 \begin{align}
    \rys'>& I(T_1;S_1)\label{eq:covsourcenc1}\\
    \ryd'>& I(T_2;S_1)\label{eq:covsourcenc2}\\
    \rys'+\ryd'>& I(T_1;S_1)+I(T_2;S_1)+I(T_1;T_2|S_1)\label{eq:covsourcenc3}
\end{align}

Note that  \eqref{eq:covsourcenc1}-\eqref{eq:covsourcenc3} can be showed using techniques similar to that in  \cite{ElgaVan81}-\cite{ElgaKim10} (mutual covering lemma). 

Then, the source 
transmits i.i.d $x_1^n(w_{13,i},w_{12,i})$ drawn according to $p(x_1|t_1,t_2,s_1)$.

2) Knowing $s_2^n(i)$, the relay searches for the smallest $m\in[1,2^{nR_{23}'}]$ such that
$(k_2^n(w_{23,i},m),s_2^n(i))\in A_\epsilon^n$. Denote this $m$ as $m_i$. If no such index $m$  exists, an encoding error is declared. Based on covering lemma, for sufficiently large $n$ such an index $m_i$ can be found with arbitrarily high probability, if
\begin{equation}\label{eq:GPCFrelaync1}
R'_{23}\geq I(K_2;S_2).
\end{equation} 

3) At the relay, assume 
$(\hat y_2^n(z_{i-1}|t_{i-1},r_{i-1},w_{23,i-1},m_{i-1}),\\k_2^n(w_{23,i-1},m_{i-1}),q_2^n(t_{i-1},r_{i-1}|w_{23,i-1},m_{i-1}),
y_2^n(i-1),\\t_2^n(w_{12,i-1},l_{i-1}),s_2^n(i-1))\in A_\epsilon^n$, 
and $z_{i-1} \in B(t_i)$. Knowing
$t_i$, $s_2^n(i)$, $w_{23,i}$ and $m_i$, the relay looks for the smallest $r\in[1,2^{nR'_2}]$ denoted as $r_i$ such that $(q_2^n(t_i,r_i|w_{23,i},m_i),k_2^n(w_{23,i},m_i),s_2^n(i))\in A_\epsilon^n$.
 For sufficiently large $n$  there exists such an index $r_i$ with arbitrarily high probability, if
\begin{equation}\label{eq:GPCFrelaync2}
R_{2}'>I(Q_2;S_2|K_2).
\end{equation}

Then the relay  transmits i.i.d $x_2^n(t_i|w_{23,i})$  drawn according to $p(x_2|q_2,k_2,s_2)$.

\textit{Decoding (at the end of block $i$):} The destination  at the end of block $i$  decodes $w_{13,i-1}$, $w_{23,i}$, and the relay  decodes
$w_{12,i}$.

1) Knowing $t_i$ from the previous block, the relay seeks  a unique pair $(\hat w_{12,i}, \hat l_i)$ such that $(t_2^n(\hat w_{12,i},\hat l_i),  y_2^n(i),q_2^n(t_i,r_i|w_{23,i},m_i),k_2^n(w_{23,i},m_{i}),s_2^n(i))$ $\in A_\epsilon^n$. For sufficiently large $n$, $(\hat w_{12,i}, \hat l_i)=(w_{12,i},l_i)$ with arbitrarily small probability of error if
\begin{equation}\label{eq:sourcerelayCFsourcenc}
\ryd+\ryd'<I(T_2;Y_2,S_2|K_2,Q_2).
\end{equation}

Note that we have performed a full decoding of the message $w_{12}$  and   the index $l$, at the relay. So, the full vector $t_2^n(w_{12,i},l_i)$ has been decoded.

2) The relay finds a unique index $z$ such that: 
$(\hat y_2^n(z|t_i,r_i,w_{23,i},m_i),q_2^n(t_{i},r_{i}|w_{23,i},m_i),k_2^n(w_{23,i},m_{i}), y_2^n(i),$ $t_2^n(w_{12,i},l_i),s_2^n(i))\in A_\epsilon^n$.
There exists such an index $z$ with arbitrarily high probability, if $n$ is
sufficiently large and
\begin{equation}\label{eq:WZCFrelaync}
\hat R_2> I(\hat Y_2;Y_2,S_2,T_2|K_2,Q_2).
\end{equation}

3) At first, the destination finds a unique pair $(\hat w_{23,i}, \hat m_i)$, such that $(k_2^n(\hat w_{23,i},\hat m_{i}),y_3^n(i))\in A_\epsilon^n$.
The decoding error can be made small if
\begin{equation}\label{eq:relaydestCFrelaync}
R_{23}+R_{23}'<I(K_2;Y_3).
\end{equation}

4) Then, the destination looks for a unique pair $(\hat t_i,\hat r_i)$ such that
$(q_2^n(\hat t_i,\hat r_i|w_{23,i},m_i),k_2^n(w_{23,i},m_{i}),y_3^n(i))\in A_\epsilon^n$. 
The decoding error can be made small if
\begin{equation}\label{eq:relaydestCFrelaync2}
   R_2+R'_2 < I(Q_2;Y_3|K_2).
\end{equation}

Note that we have performed a full decoding of $k_2^n$  and $q_2^n$ at the destination in  steps 3 and 4.

5) Now, knowing  $t_{i-1},r_{i-1},w_{23,i-1}$ and $m_{i-1}$ (from the previous block), the destination calculates a set of indices $z$ denoted by the list $\Lc(y_3^n(i-1))$  such that
$(\hat y_2^n(z|t_{i-1},r_{i-1},w_{23,i-1},m_{i-1}),q_2^n(t_{i-1},r_{i-1}|w_{23,i-1},m_{i-1}),\\k_2^n(w_{23,i-1},m_{i-1}),y_3^n(i-1)) \in A_\epsilon^n$.
Then the destination declares that $\hat z_{i-1}$ has been sent in block
$i-1$, if $\hat z_{i-1} \in B(t_i) \cap \Lc (y_3^n(i-1))$.
With arbitrarily high probability $\hat z_{i-1}=z_{i-1}$,  if
$n$ is sufficiently large and
\begin{equation}\label{eq:listcodeCFrelaync}
    \hat R_2 < R_2 + I(\hat Y_2;Y_3|Q_2,K_2).
\end{equation}

6) Finally, the destination uses  $y_3^n(i-1)$ and $\hat y_2^n(z_{i-1}|t_{i-1},r_{i-1},w_{23,i-1},m_{i-1})$,
 and declares that $\hat w_{13,i-1}$ is sent, if there is a unique $\hat w_{13,i-1}$ for some $k_{i-1} \in[1,2^{n\rys'}]$  such that
$(t_1^n(\hat w_{13,i-1},k_{i-1}),q_2^n(t_{i-1},r_{i-1}|w_{23,i-1},m_{i-1}),y_3^n(i-1),k_2^n(w_{23,i-1},m_{i-1}),\hat y_2^n(z_{i-1}|t_{i-1},r_{i-1},w_{23,i-1},m_{i-1}))\in A_\epsilon^n$.
Thus,  $\hat w_{13,i-1}=w_{13,i-1}$ with arbitrarily high probability, if $n$ is sufficiently large and
\begin{equation}\label{eq:destCFsourcenc}
    R_{13}+R_{13}'<I(T_1;Y_3,\hat Y_2|Q_2,K_2).
\end{equation}

Combining \eqref{eq:covsourcenc1},  \eqref{eq:covsourcenc2}, \eqref{eq:covsourcenc3} with \eqref{eq:sourcerelayCFsourcenc}, \eqref{eq:destCFsourcenc}, and  \eqref{eq:GPCFrelaync1} with \eqref{eq:relaydestCFrelaync} yields \eqref{eq:lowerCFsourcerelaync1}-\eqref{eq:lowerCFsourcerelaync4}.  The constraint in \eqref{eq:condCFsourcerelaync} follows from combining \eqref{eq:GPCFrelaync2}, \eqref{eq:WZCFrelaync}, \eqref{eq:relaydestCFrelaync2}  and \eqref{eq:listcodeCFrelaync}. Moreover, note that since $R_{13},R_{12},R_{23}>0$ and due to \eqref{eq:condCFsourcerelaync}, full decoding of $t_2^n,k_2^n$  and $q_2^n$ does not cause additional constraints. This completes the proof. \IEEEQED

\begin{remark}\label{mk:CFsourcerelaync2}
By setting $S=S_1=S_2=\emptyset$ (to make state-independent channel),  and re-defining $T_1=U_1$, $T_2=U_2$, $K_2=V$ and $Q_2=X_2$  in  the region $\Rc_1$, it yields the achievable region for RCPM in \cite[Theorem 2]{TannNos07}. After these substitutions and simplifications, the only difference is in  \eqref{eq:condCFsourcerelaync}, with its counterpart in \cite[Equation 21]{TannNos07}, where the reason is as follows: since the destination does not decode $U_2$ ($T_2$ is re-defined as $U_2$), in the random coding process $\hat Y_2$ is not generated for each $U_2$. Therefore, due to covering lemma \cite{ElgaKim10}, $\hat Y_2$ should cover both $Y_2$ and $U_2$, and  $\hat R_2>I(\hat Y_2;Y_2,U_2|V,X_2)$ is obtained.  But in \cite{TannNos07} the authors did not notice this point,  and to obtain $\hat R_2$ they erroneously made conditioning on $U_2$ and so they obtained $\hat R_2>I(\hat Y_2;Y_2|U_2,V,X_2)$. 
\end{remark}
\begin{remark}\label{mk:CFsourcerelaync3}
By setting $(S_2,Q_2,K_2,X_2,\hat Y_2)=\emptyset$ in   region $\Rc_1$ (i.e., disable relaying), $\Rc_1$ reduces  to  the achievable rate region  for BC with non-causal CSIT that has been derived in \cite{SteiSha05}. 
\end{remark} 
\begin{remark}\label{mk:CFsourcerelaync4}
The region $\Rc_1$, reduces to the  achievable rate for SD-RC  with non-causal perfect CSI only at the source in \cite[Theorem 1]{AkhbMirAre09} by setting $K_2=\emptyset$ (no private $W_{23}$ from the relay to  destination and $R_{23}=0$), $T_2=\emptyset$ (no private $W_{12}$ from the source to the relay and $R_{12}=0$), $S_2=\emptyset$ and $S_1=S$  (only the source is informed of perfect CSI) and re-defining $Q_2=X_2$ in $\Rc_1$. Moreover, setting $S_1=\emptyset$, $S_2=S$, $T_2=\emptyset$,   $K_2=\emptyset$,  and re-defining {\small $T_1=X_1$} in $\Rc_1$, results in the achievable rate for SD-RC  with non-causal perfect CSI only at the relay in \cite[Theorem 2]{AkhbMirAre09}.
\end{remark}

Now, we specialize Theorem \ref{th:noncausal} to  cases where perfect CSI is available non-causally only at the source, only at the relay, or both at the source and relay.
\begin{corol}\label{corol:noncausal}
Theorem \ref{th:noncausal} is specialized  to an  achievable rate region for SD-RCPM  with non-causal perfect CSI only at the source, by setting   $S_1=S$, $S_2=\emptyset$ (since only the source is informed of CSI) and re-defining  $K_2=V$, $Q_2=X_2$ in $\Rc_1$.  Furthermore, setting $S_1=\emptyset$,  $S_2=S$ and re-defining $T_1=U_1$, $T_2=U_2$   in $\Rc_1$,  yields an achievable rate region for SD-RCPM  with non-causal perfect CSI only at the relay. Setting $S_1 = S_2 = S$ in $\Rc_1$ yields an achievable rate region for SD-RCPM  with non-causal perfect CSI at both the source and relay.\end{corol}
\vspace{-2mm}
\section{SD-RCPM with Causal CSI}\label{sec:causal}
In many practical applications, the state sequence is not known in advance, and has to be known in a causal manner. In this section, we consider SD-RCPM with partial causal CSI  at the source and the relay. In the following, the achievable rate region  for this channel is derived.
\begin{theorem}\label{th:causal}
In the discrete memoryless SD-RCPM with causal CSI $S_1$ and $S_2$ respectively  at the source and  the relay, the nonnegative rate tuples  $(R_{12},R_{23},R_{13})$ denoted as $\Rc_2$, satisfying:
\begin{align}
R_{13}<&I(T_1;\hat Y_2,Y_3|K_2,Q_2)\label{eq:lowerCFsourcerelayc1}\\
R_{12}< &I(T_2;Y_2,S_2|K_2,Q_2)\label{eq:lowerCFsourcerelayc2}\\
R_{13}+R_{12}< &I(T_1;\hat Y_2,Y_3|K_2,Q_2)\nonumber\\
&+I(T_2;Y_2,S_2|K_2,Q_2)-I(T_1;T_2)\label{eq:lowerCFsourcerelayc3}\\
R_{23}<&I(K_2;Y_3)\label{eq:lowerCFsourcerelayc4},\\
\!\!\!\!\!\textrm{subject to}\ \ I(\hat Y_2&;Y_2,S_2,T_2|K_2,Q_2,Y_3) \leq I(Q_2;Y_3|K_2),\label{eq:condCFsourcerelayc}
\end{align} are achievable for any joint p.m.f of the form
\begin{align}
&p(s,s_1,s_2)p(k_2)p(q_2|k_2)p(x_2|q_2,k_2,s_2)p(t_1,t_2)p(x_1|t_1,t_2,s_1)\nonumber\\
&\ \times p(y_2,y_3|x_1,x_2,s)p(\hat y_2|y_2,q_2,k_2,s_2,t_2),\label{eq:pmfsourcerelayc}
\end{align} where $X_1=f_1(T_1,T_2,S_1)$, $X_2=f_2(Q_2,K_2,S_2)$ and $f_1(\cdot)$ and $f_2(\cdot)$ are two arbitrary deterministic functions.
\end{theorem}
\begin{proof}
Similar to the single user channel, the proof of the causal case follows the lines of the proof  for the non-causal case in Theorem \ref{th:noncausal}. The only difference is that since in the causal case  $(T_1,T_2)$ and $S_1$ are independent (see \eqref{eq:pmfsourcerelayc}), the encoding scheme for the non-causal case is reduced to the one that does not include GP binning, and therefore does not require non-causal knowledge of $S_1$. Similar situation happens for $(K_2,Q_2)$ and $S_2$. So, the coding scheme is obtained based on using Shannon's strategy \cite{Shan58} at the source and the relay which are informed of CSI to incorporate the state knowledge,
using random binning corresponds to Marton's simplified region for BC \cite{ElgaVan81} at the source, and using CF relaying scheme at the relay. The proof is rather straightforward and  is omitted here for brevity.
\end{proof}
\begin{remark}\label{mk:CFsourcec1}
The expression of  $\Rc_2$ can be interpreted as a special case of $\Rc_1$, where $(T_1,T_2)$ are independent of $S_1$, and $(K_2,Q_2)$ are independent of $S_2$. This is similar to the relation between the expression for the capacity of state-dependent single user channel with causal CSI \cite{Shan58}, and its non-causal counterpart \cite{GelfPin80}. Hence, similar to Remark \ref{mk:CFsourcerelaync4},  $\Rc_2$ can  be reduced to  achievable rates derived for SD-RC with causal CSI in \cite{MirmAkhAre09}. Moreover, similar to  Corollary \ref{corol:noncausal},   $\Rc_2$ can be specialized to  achievable rate regions for SD-RCPM for the cases where CSI is available causally only at the source, only at the relay, or both at the source and relay.
\end{remark}
\vspace{-2mm}
\section{Gaussian Example}\label{sec:Gaussian}
In this section, we consider a general full-duplex Gaussian RC with private messages and with  additive independent Gaussian state and noise, and we refer to it as Gaussian SD-RCPM.  The outputs at the relay and the destination at time $j=1,\ldots,n$ for Gaussian SD-RCPM are given by:
{\small
\begin{equation}\label{eq:Gaussmodel2}
Y_{2j}=X_{1j}+Z_{2j}+S_j\ \  \textrm{and}\ \  Y_{3j}=X_{1j}+X_{2j}+Z_{3j}+S_j,
\end{equation}}where  $X_{1j}$ and $X_{2j}$ are transmitted signals by the source and the relay with individual average power constraints  $P_1$ and $P_2$. $Z_{2j}$ and $Z_{3j}$ are  independent zero-mean Gaussian RVs with variances   $N_2$ and $N_3$. The channel state $S_j$  is a zero-mean Gaussian RV with variance $Q$, and is  independent of $Z_{2j}$ and $Z_{3j}$. 

As an example, we consider the case where the source knows state sequence $S^n$ non-causally and perfectly, while the relay is not informed of CSI. To provide an achievable rate region for this Gaussian case, we use the results obtained  in Corollary \ref{corol:noncausal} for this scenario.  For simplicity, we assume that input distributions are Gaussian, although this assumption may not be optimal. Similar to \cite{TannNos07},  we assume that the generated codebook at the source is mapped into Gaussian RVs as $X_1=U_1+U_2$, and $\rho$ denotes the correlation coefficient between $U_1$ and $U_2$. Power constraint $P_1$ at the source, and parameter $\rho$ lead to power constraints for $U_1$ and $U_2$ (i.e., $P_{u2}=\gamma P_1, P_{u1}+P_{u2}+2\rho\sqrt{P_{u1}P_{u2}}=P_1$). Now, since the source is informed of CSI, it uses Dirty Paper Coding (DPC) \cite{Cost83} for  transmission  of message $W_{12}$ (codeword $U_2$) to relay, and uses  another  DPC for transmission  of message $W_{13}$ (codeword $U_1$) to destination. Hence, similar to Costa's initial DPC, the auxiliary RVs $T_1$ and $T_2$ are defined as $T_1=U_1+\alpha_1 S$ and $T_2=U_2+\alpha_2 S$, respectively. Moreover, to partially cancel the state,  arbitrary correlations are assumed between $U_1$ and $S$ (via $\rho_{u_1s}$), and between $U_2$ and $S$ (via $\rho_{u_2s}$) which is called Generalized DPC (GDPC) \cite{ZaidVan08b}. Without loss of generality $\rho_{u2s}$ can be set to zero (i.e., using DPC instead) since $U_2$ is defined for transmission of $W_{12}$
on the point to point link between the source and relay. The relay is not informed of CSI, so $E(X_2S)=0$ and it does not use DPC. At the relay, $\theta$ is the fraction of the relay's power which is dedicated to the relay's private message. Therefore, $X_2=V+X'_2$, where $V\sim\mathcal{N}(0,\theta P_2)$ and $X'_2\sim\mathcal{N}(0,(1-\theta)P_2)$ are independent. 
We assume $\hat Y_2=\beta Y_2+f T_2+\hat Z$ where $\beta\in[0,1]$ and $f\in[-1,1]$, and compression noise $\hat Z\sim\mathcal{N}(0,\hat N)$ is independent of $S,X_1,X_2,Z_2,Z_3$. We  let $f$ to be negative to consider the cases   where the relay peels-off the part intended for the relay (i.e., $T_2$), and then compresses the remaining part to send to the destination. $f=0$ refers to the case where  relay only compresses $Y_2$ without removing or compressing $T_2$. So, we evaluate \eqref{eq:lowerCFsourcerelaync1}-\eqref{eq:condCFsourcerelaync}  for this scenario (informed source only) and for a given ($\rho,\gamma,\alpha_1,\alpha_2,\rho_{u_1s},\rho_{u2s},\theta,\beta,f$). By varying these parameters, and taking the union of the resulted regions, the achievable rate region for this scenario is established (its expression is omitted here due to space). By ignoring private messages in this channel (i.e., setting $\gamma=0$ and $\theta=0$), this achievable region  reduces to  the achievable rate  derived in \cite[Theorem 4]{AkhbMirAre09} for Gaussian SD-RC   with non-causal CSI only at the source.

In Fig. \ref{fig:CFRCPMsource1}, our  achievable rate region for Gaussian SD-RCPM with non-causal CSI at the source, is plotted for 
$P_1=Q=N_3=10\textrm{dB}$, $P_2=15\textrm{dB}$, $N_2=0\textrm{dB}$,  parameterized by the fraction of the relay's power dedicated to relay's private message ($\theta$), and we can see the trade-off between $R_{13}$ and $R_{12}$. As illustrated in Fig.~\ref{fig:CFRCPMsource1}, when  $\theta$ is increased, the maximum achievable $R_{13}$ is decreased. For example, $\theta=0$ refers to the case where the relay sends no private message to the destination ($W_{23}=\emptyset$), and Gaussian SD-RCPM reduces to Gaussian SD-PCRBC. Therefore, for this case  we can compare our CF based  achievable  region, with the DF based  achievable region   derived in \cite[Lemma 2]{ZaidVan07} for  SD-PCRBC with informed source only. Moreover, the outer bound on the capacity region of Gaussian PC-RBC (state-independent) in \cite[Theorem 5]{LianVee07} is a trivial  outer bound on the capacity region of SD-PCRBC. So  we illustrate this outer bound, DF and CF based rate regions  in Fig. \ref{fig:Figure2} for two examples of noise configuration: 1) $N_2=N_3$ ($N_2,N_3=10\textrm{dB}, P_1=20\textrm{dB}$) which is denoted as ``power set 1'' in Fig. \ref{fig:Figure2}, 2) $N_3>N_2$ ($N_3=10\textrm{dB}, N_2=8\textrm{dB}, P_1=16\textrm{dB}$) denoted as ``power set 2''. For both ``power set 1'' and ``power set 2'', $Q=10\textrm{dB}$, $P_2=25\textrm{dB}$. For these power sets, CF outperforms DF in PC-RBC (state-independent). So,  in the state-dependent version of this channel it can outperform DF which is shown in Fig. \ref{fig:Figure2}. In these examples, if $P_2$ is also increased, CF based bound becomes closer to its respective outer bound. Note that for noise configuration where $N_3<N_2$, it is clear that CF outperforms DF, and  to  prevent  having a messy figure this case is not depicted in Fig. \ref{fig:Figure2}. We remark that  these power sets are chosen to have cases that can be presented  in one figure. It is also interesting to note that for all these cases,  maximum achievable $R_{12}$ (when $R_{13}=0$) coincides with that of derived for PC-RBC (state-independent), and this confirms that a complete state cancellation is performed for sending $U_2$ on the source-relay link for $R_{13}=0$.

\begin{figure}[tb]
  \centering
  \includegraphics[width=6.4cm]{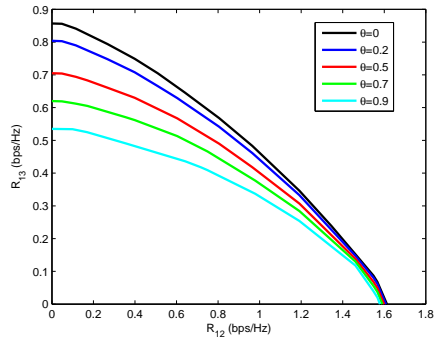}
    \caption{CF based achievable rate region for general Gaussian SD-RCPM with non-causal CSI only at the source, parametrized by $\theta$.}
    \label{fig:CFRCPMsource1}
    \vspace{-4mm}
\end{figure}

\begin{figure}[tb]
  \centering
  \includegraphics[width=6.4cm]{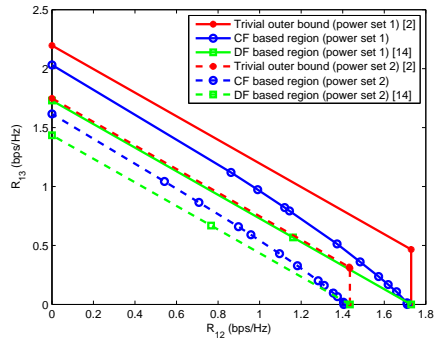}
    \caption{General Gaussian SD-PCRBC  with  non-causal CSI only at the source, for ``power set 1'' and ``power set 2''.}
    \label{fig:Figure2}
     \vspace{-5mm}
\end{figure}

\vspace{-2mm}
\section{Conclusion}\label{sec:conclusion}
We introduced SD-RCPM as a generalization of SD-RC and SD-PCRBC. In order to have a unified view, both causal and non-causal cases were investigated. In each case, an achievable rate region for SD-RCPM with partial CSI at the source and relay was derived  using CF scheme. We also derived an achievable rate region for Gaussian version of SD-RCPM with non-causal perfect CSI only at the source. In ongoing work, we are using  DF scheme for both  discrete memoryless and Gaussian version of SD-RCPM with causal and non-causal CSI.

\end{document}